\documentclass[submission,copyright,creativecommons]{eptcs}
\providecommand{\event}{LSFA 2012} 
\usepackage{amsmath,amssymb,stmaryrd}
\usepackage{enumerate}
\usepackage[english]{babel}
\usepackage[latin1]{inputenc}

\usepackage{proof}
\usepackage{bussproofs}
\usepackage{cite}
\usepackage{algpseudocode}
\usepackage{algorithm}
\usepackage{textcomp}
\usepackage{color}
\usepackage{proof}
\usepackage{graphics}
\usepackage{breakurl,url}    
\usepackage{amsthm}

\newcommand{\blind}{\textsf{blind}}
\newcommand{\sign}{\textsf{sign}}

\newcommand{\ep}{\textsf{EP}}
\newtheorem{definition}{Definition}
\newtheorem{theorem}{Theorem}
\newtheorem{lemma}{Lemma}
\newtheorem{remark}{Remark}
\newtheorem{corollary}{Corollary}
\newtheorem{example}{Example}

\title{Elementary Deduction Problem for Locally Stable Theories with Normal Forms\thanks{Work supported by grants from the CNPq/CAPES \emph{Science without Borders}  programme and FAPDF PRONEX.}}
\author{Mauricio Ayala-Rinc\'{o}n \thanks{Author partially supported by CNPq.}
\institute{Departamentos de Matem\'{a}tica e\\
Computa\c{c}\~{a}o \\ Grupo de Teoria da Computa\c{c}\~{a}o}
\institute{Universidade de Bras\'{i}lia, Brazil}
\email{ayala@unb.br}
\and
Maribel Fern\'{a}ndez
\institute{Department of Informatics}
\institute{King's College London,  UK}
\email{maribel.fernandez@kcl.ac.uk}
\and
Daniele Nantes-Sobrinho\thanks{Corresponding author. Author supported by CNPq}
\institute{Departamento de Matem\'{a}tica \\ Grupo de Teoria da Computa\c{c}\~{a}o}
\institute{Universidade de Bras\'{i}lia, Brazil}
\email{dnantes@mat.unb.br}
}
\def\titlerunning{Elementary Deduction for  Locally Stable Theories Combined with Normal Forms}
\def\authorrunning{M. Ayala-Rinc\'{o}n, M. Fern\'{a}ndez \& D. Nantes-Sobrinho}
\begin{document}
\maketitle

\begin{abstract}
We present an algorithm to decide the intruder deduction problem (IDP) for a class of locally stable theories enriched with normal forms.  Our result relies on a new and efficient algorithm to solve a restricted case of higher-order associative-commutative matching, obtained by combining the \emph{Distinct Occurrences of AC-matching} algorithm and a standard algorithm to solve systems of linear Diophantine equations. A translation between natural deduction and sequent calculus allows us to use the same approach to decide the \emph{elementary deduction problem} for locally stable theories. As an application, we model the theory of blind signatures and derive an algorithm to decide IDP in this context,  extending previous decidability results.
\end{abstract}

\section*{Introduction}
There are different approaches to model cryptographic protocols and to analyse their security properties~\cite{survey}. One technique consists of proving that an attack requires solving an algorithmically hard problem; another consists of using a process calculus, such as the spi-calculus~\cite{spi},  to represent the operations performed by the participants and the attacker. In recent years, the deductive approach of Dolev and Yao~\cite{dolev}, which abstracts from algorithmic details and models an attacker by a deduction system, has successfully shown the existence of flaws in well-known protocols.  A  deduction system under Dolev-Yao's approach specifies how the attacker can obtain new information from previous knowledge  obtained either by eavesdropping the communication between honest protocol participants (in the case of a passive attacker), or by eavesdropping and fraudulently emitting messages (in the case of an active attacker).  The \emph{intruder deduction problem} (IDP) is the question of whether a passive eavesdropper can obtain a certain information from messages observed on the network.

Abadi and Cortier's approach~\cite{AbadiCo2006} proposes conditions for analysing message deducibility and indistinguishability relations for security protocols  modelled in the applied pi-calculus~\cite{appliedpicalculus}. In particular,~\cite{AbadiCo2006} shows that IDP is decidable for \emph{locally stable} theories. However, to ensure the soundness of this approach, the definition of locally stable theories given in~\cite{AbadiCo2006} needs to be modified (as confirmed via personal communication with the second author of~\cite{AbadiCo2006}). In this work, we made the necessary modifications and propose a new approach to solve IDP in the context of locally stable theories.

Our notion of locally stable theory is based on the existence of a finite and computable saturated set, but, unlike~\cite{AbadiCo2006}, our saturated sets include normal forms\footnote{With this simple modification, the correctness proof in~\cite{AbadiCo2006} can also be carried out, fixing a gap in Lemma 11.}.  The new approach we propose in order to prove the decidability of IDP is based on an algorithm to solve a restricted case of higher-order associative-commutative matching (AC-matching). To design this algorithm  we use  well-known results for solving  systems of linear Diophantine equations (SLDE)~\cite{BouConDe,ClauFor,frumkin,papadimitriou}, which we  combine with a polynomial algorithm to solve the DO-ACM problem (Distinct Occurrences of AC-Matching)~\cite{narendran}.  

In the case where the signature of the equational theory contains, for each AC function symbol $\oplus$, its corresponding inverse $i_{\oplus}$, we obtain a decidability result which is polynomial with relation to the size of the saturated set (built from the initial knowledge of the intruder). 
Thanks to the use of the algorithm for solving SLDE over $\mathbb{Z}$, we avoid an exponential time search over the solution space in the case of AC symbols (improving over~\cite{AbadiCo2006}, where an exponential number of possible combinations have to be considered). For more details we refer the reader to the extended version of this paper \cite{AFSextended}.

After introducing the class of locally stable theories and proving the decidability of the IDP for protocols in this class, we show that the Elementary Deduction Problem (EDP) introduced in \cite{TiGo2009} is also decidable in polynomial time with relation to the size of a saturated set of terms. EDP is stated as follows: given a set $\Gamma$ of messages  and a message $M$, is there an $E$-context $C[\ldots]$ and messages $M_1,\ldots, M_k\in \Gamma$ such that $C[M_1,\ldots, M_k]\approx_{E}M$? Here, $E$ is the equational theory modelling the protocol. We use this approach to model theories with blind signatures. As an application, using a previous result that links the decidability of the EDP to the decidability of the IDP when the theory $E$ satisfies certain conditions, we obtain decidability of IDP for a subclass of locally stable theories combined with the theory $B$ of blind signatures. In this way, we generalise a result from~\cite{AbadiCo2006} (Section 5.2.4): it is not necessary to prove that the combination of the theories $E$ and $B$ is locally stable.

\textbf{Related Work.} The analysis of cryptographic protocols has attracted a lot of attention in the last years and several tools are available to try to identify possible attacks, see Maude-NPA~\cite{maudeNPA07}, ProVerif~\cite{proverif}, CryptoVerif~\cite{cryptoverif}, Avispa~\cite{avispa}, Yapa~\cite{YAPA}.  

Sequent calculus formulations of Dolev Yao intruders~\cite{Tiu2007} have been used in a formulation of open bisimulation for the spi-calculus. In~\cite{TiGo2009}, deductive techniques for dealing with a protocol with blind signatures in mutually disjoint AC-convergent equational theories, containing a unique AC operator each, are considered. As an alternative approach, the intruder's deduction capability is modelled inside a sequent calculus modulo a rewriting system, following the approach of~\cite{BeC06}. Then, the IDP is reduced in polynomial time to EDP.

By combining the techniques in~\cite{TiGo2009} and~\cite{Bursucconstraints},  the IDP formulation for an Electronic Purse Protocol with blind signatures was proved to reduce in polynomial time to EDP for an AC-convergent theory containing three different $AC$ operators and rules for exponentiation~\cite{nantesayala}, extending the previous results. However, no algorithm was provided to decide EDP.
More precisely, assuming that EDP is solved in time $O(f(n))$,  it was proved that IDP reduces polynomially to EDP with complexity $O(n^k \times f(n))$, for some constant $k$. Thus, whenever the former problem is polynomial, the IDP is also polynomial.

\textbf{Contributions.} We present a technique to decide EDP or IDP in
AC-convergent equational theories. Our approach is based on a ``local
stability'' property inspired by~\cite{AbadiCo2006}, instead of
proving that the deduction rules are ``local'' in the sense
of~\cite{mcallester} as done in many previous works~\cite{CoLuSh03,
  De2006, Laf07, Bursucconstraints}.  More precisely, the main
contributions of this paper are:
\begin{itemize}
\item We adapt and refine the technique proposed in~\cite{AbadiCo2006}, where 
deducibility and indistinguishability relations are claimed to be decidable 
in polynomial time for  locally stable theories. 
First, we changed the definition of locally stable theories,
adding normal forms, which are needed to carry out the decidability proofs.
Second, we designed a new algorithm to decide IDP in locally stable theories.
The algorithm provided in~\cite{AbadiCo2006} is polynomial for the class
of subterm theories (Proposition 10 in~\cite{AbadiCo2006}), but the proof
does not extend directly to locally stable theories (despite the statement
in Proposition 16). Our algorithm relies on solving a restricted case of  higher-order
AC-matching problem that is used to decide the deduction relation. It
is a combination of two standard algorithms: one for solving the
DO-ACM problem \cite{narendran} which has a
polynomial bound in our case; and one for solving systems of Linear
Diophantine Equations(SLDE), which is polynomial in $\mathbb{Z}$~\cite{BouConDe,ClauFor,frumkin,papadimitriou}. Using this algorithm we prove that IDP is decidable in polynomial time with respect to the saturated set of terms, for locally stable theories with inverses.
\item A decidability result for the EDP for locally stable theories, which extends the work of Tiu and Gor\'{e}~\cite{TiGo2009}. As an application, we present a strategy to decide IDP for locally stable theories combined with blind signatures. Here, the combination of theories does not need to be locally stable.
\end{itemize}

In order to get the polynomial decidability result claimed in~\cite{AbadiCo2006} for
locally stable theories,  we had to restrict to theories that contain, for each $AC$
symbol in the signature, the corresponding inverse. The inverses are
necessary when we interpret our term algebra inside the integers $\mathbb{Z}$
to solve SLDE (terms headed by the inverse function will be seen as
negative integers). If the theory does not contain inverses, we would
have to solve the SLDE for $\mathbb{N}$  which is a well known NP-complete problem.

\section{Preliminaries}

Standard rewriting notation and notions are used (e.g. \cite{baader}).  We assume the following sets: a countably infinite set $N$ of \emph{names} (we use $a,b,c, m$ to denote names); a countably infinite set $X$ of  \emph{variables}  (we use  $x,y,z$ to denote variables); and a finite \emph{signature} $\Sigma$, consisting of function names and their arities. We write $arity(f)$ for the arity of a function $f$, and let $ar(\Sigma)$ be the maximal arity of a function symbol in $\Sigma$.

The set of \emph{terms} is generated by the following grammar:
$$M,N:= a\,|\, x\,|\, f(M_1,\ldots, M_n)$$
where $f$ ranges over the function symbols of $\Sigma$ and $n$ matches the arity of $f$, $a$ denotes a name in $N$ (representing principal names, nonces, keys, constants involved in the protocol, etc) and $x$ a variable. We denote by $V(M)$ the set of variables occurring in $M$. A message $M$ is \emph{ground} if $V(M)=\emptyset$. The \emph{size} $|M|$ of a term $M$ is
defined by $|u|=1$, if $u$ is a name or a variable; and $|f(M_1,\ldots, M_n)|=1+\sum_{i=1}^n|M_i|$.

The set of \emph{positions} of a term $M$, denoted by $\mathcal{P}os(M)$, is defined by $\mathcal{P}os(M):= \{\epsilon\}$, if $M$ is a name or a variable; and  $\mathcal{P}os(M):= \{\epsilon\} \cup \bigcup_{i=1}^n\{ip \,|\, p \in \mathcal{P}os(M_i)\}$, if $M=f(M_1, \ldots, M_n)$ where $f\in \Sigma$. 
The position $\epsilon$ is called the \emph{root} position.  The size of $|M|$ coincides with the cardinality of $\mathcal{P}os(M)$. The set of \emph{subterms} of $M$ is defined as $st(M)=\{M|_p \,|\, p \in \mathcal{P}os(M)\}$, where $M|_p$ denotes the subterm of $M$ at position $p$. For a set  $\Gamma$ of terms, the notion of subterm can be extended as usual: $st(\Gamma):= \bigcup_{M\in \Gamma}st(M)$. For $p \in \mathcal{P}os(M)$, we denote by $M[t]_p$ the term that is obtained from $M$ by replacing the subterm at position $p$ by $t$.

A term rewriting system (TRS) is a set $\mathcal{R}$ of oriented equations over terms in a given signature. For terms $s$ and $t$, $s\rightarrow_{\mathcal{R}} t$ denotes that $s$ rewrites to $t$ using an instance of a rewriting rule in $\mathcal{R}$. The transitive, reflexive-transitive and equivalence closures of $\rightarrow_{\mathcal{R}}$ are denoted by  $\stackrel{+}{\rightarrow}_{\mathcal{R}}, \stackrel{*}{\rightarrow}_{\mathcal{R}}$ and $\stackrel{*}{\leftrightarrow}_{\mathcal{R}}$, respectively. 
The equivalence closure of the rewriting relation, $\stackrel{*}{\leftrightarrow}_{\mathcal{R}}$, is 
denoted by $\approx_{\mathcal{R}}$. 

Given a TRS $\mathcal{R}$ in which some function symbols are assumed to be AC, and  two terms $s$ and $t$, $s\rightarrow_{\mathcal{R}\cup AC}t$ if there exists $w$ such that $s=_{AC}w$ and $w\rightarrow_{\mathcal{R}} t$, where $=_{AC}$ denotes equality modulo AC (according to the AC assumption on function symbols). For every term $s$, the set of normal forms $s\downarrow_{\mathcal{R}}$ (closed modulo AC) of $s$ is the set of terms $t$ such that $s\stackrel{*}{\rightarrow}_{\mathcal{R}\cup AC}t$ and $t$ is irreducible for $\rightarrow_{\mathcal{R}\cup AC}$. $\mathcal{R}$ is said to be AC-convergent whenever it is AC-terminating and AC-confluent.

We equip the signature $\Sigma$ with an equational theory $\approx_E$  induced by a set of $\Sigma$-equations $E$, that is,  $\approx_E $ is the smallest equivalence relation that contains $E$ and is closed under substitutions and compatible with $\Sigma$-contexts.  An equational theory $\approx_E$ is said to be equivalent to a TRS $\mathcal{R}$ whenever $\approx_{\mathcal{R}}\; =\; \approx_E$.  An equational theory $\approx_E$ is AC-convergent when it has an equivalent rewrite system $\mathcal{R}$ which is AC-convergent. In the next sections, given an AC-convergent equational theory $\approx_E$, normal forms of terms are computed with respect to the TRS $\mathcal{R}$ associated to $\approx_E$, unless otherwise specified. To simplify the notation we will denote by $E$ the  equational theory induced by the set of $\Sigma$-equations $E$. We will denote by $\Sigma_E$ the signature used in the set of equations  $E$. 
The \emph{size} $c_E$ of an equational theory $E$ with an associated TRS $\mathcal{R}$ consisting of rules $\bigcup_{i=1}^k\{l_i\rightarrow r_i\}$ is defined as $c_E=max_{1\leq i\leq k}\{|l_i|,|r_i|, ar(\Sigma)+1\}$. For $\mathcal{R}= \emptyset$, define $ c_E= ar(\Sigma)+1$. 

Let $\square$ be a new symbol which does not yet occur in $\Sigma\cup X$. A $\Sigma$-\emph{context} is a term $t\in T(\Sigma, X\cup \{ \square\})$ and can be seen as a term with ``holes'', represented by $\square$, in it. Contexts are denoted by $C$. If $\{p_1,\ldots,p_n\}=\{p\in\mathcal{P}os(C) \,|\, C|_p = \square\}$, where $p_i$ is to the left of $p_{i+1}$ in the tree representation of $C$, then $C[T_1\ldots, T_n]:= C[T_1]_{p_1}\ldots [T_n]_{p_n}$. In what follows a context formed using only function symbols in $\Sigma_{E}$  will be called an $E$-\emph{context} to emphasize  the equational theory $E$.

 A term $M$ is said to be an $E$-\emph{alien} if $M$ is headed by a symbol $f\notin \Sigma_{E}$ or a private name/constant. We write $M==N$ to denote syntactic equality of ground terms.

In the rest of the paper, we use signatures, terms and equational theories to model protocols.  \emph{Messages} exchanged between participants of a protocol during its execution are represented by terms. Equational theories and rewriting systems are used to model the cryptographic primitives in the protocol and the algebraic capabilities of an intruder. 

\section{Deduction Problem}\label{sec:locallystable}
Given  a set $\Gamma$ that represents the information available to an attacker, we may ask whether a given ground term $M$ may be deduced from $\Gamma$ using equational reasoning. This relation is written $\Gamma \vdash M$ and axiomatised in a natural deduction like system of inference rules.

\begin{table}[ht]
\caption{System $\mathcal{N}$: a natural deduction system for intruder equational deduction}
\label{equationalreasoning}
{\small
\hrule
\begin{prooftree}
\AxiomC{$M \in \Gamma$}
\RightLabel{$(id)$}
\UnaryInfC{$\Gamma \vdash M$}
\DisplayProof
\defaultHypSeparation
\AxiomC{$\Gamma \vdash M_1 \,\,\,\ldots $}
\AxiomC{$\Gamma \vdash M_n$}
\RightLabel{$(f_I)\,f\in \Sigma_E$}
\BinaryInfC{$\Gamma \vdash f(M_1, \ldots, M_n)$}
\DisplayProof
\defaultHypSeparation
\AxiomC{$\Gamma \vdash N$}
\RightLabel{$(\approx) \,M\approx_{E} N$}
\UnaryInfC{$\Gamma \vdash M$}
\end{prooftree}
\hrule
}
\end{table}

\subsection{Locally Stable Theories}
Let $\oplus$ be an arbitrary function symbol in $\Sigma_E$ for an equational theory $E$. We write $\alpha \cdot_{\oplus} M$ for the term $M\oplus \ldots \oplus M$, $\alpha$ times ($\alpha \in \mathbb{N}$). Given a set  $S$ of terms, we write $sum_{\oplus}(S)$ for the set of arbitrary sums of terms in $S$, closed modulo $AC$:
\begin{equation*}
sum_{\oplus}(S)=\{(\alpha_1 \cdot_{\oplus}T_1)\oplus \ldots\oplus(\alpha_n \cdot_{\oplus}T_n)\,|\, \alpha_i \geq 0, 
T_i\in S\}
\end{equation*}
Define  $sum(S)= \bigcup_{i=1}^k sum_{\oplus_i}(S)$, where $\oplus_1,\ldots, \oplus_k$ are the AC-symbols of the theory.

For a rule $l\rightarrow r\in \mathcal{R}$ and a substitution $\theta$ such that
\begin{itemize}
\item either there exists a term $s_1$ such that $s=_{AC}s_1$, $s_1=_{AC}l\theta$ and $t=r\theta$;
\item or there exist terms $s_1$ and $s_2$ such that $s=_{AC}s_1 \oplus s_2$, $s_1=_{AC}l\theta$ and $t=_{AC}r\theta \oplus s_2$.
\end{itemize}
we write $s\stackrel{h}{\rightarrow}t$ and say that the reduction occurs in the head.

As in ~\cite{AbadiCo2006} we associate with each set $\Gamma$ of messages, a set of subterms in $\Gamma$ that may be deduced from $\Gamma$ by applying only ``small'' contexts.  The concept of small is arbitrary --- in the definition below, we have bound the size of an $E$-context $C$ by $c_E$ and the size of $C'$ by $c_E^2$, but other bounds may be suitable. Notice that limiting the size of an  $E$-context by $c_E$ makes the context big enough to be an instance of  any of the rules in the TRS $\mathcal{R}$ associated to $E$.

\begin{definition}[Locally Stable]\label{locallystable}
An AC-convergent equational theory $E$ is \emph{locally stable} if, for every finite set $\Gamma=\{M_1, \ldots,M_n\}$, where the terms $M_i$ are ground and in normal form, there exists a finite and computable set $sat(\Gamma)$, closed modulo $AC$, such that
\begin{enumerate}
\item $M_1, \ldots, M_n \in sat(\Gamma)$;\label{rule1}
\item if $M_1,\ldots,M_k \in sat(\Gamma)$ and $f(M_1,\ldots,M_k)\in st(sat(\Gamma))$ then $f(M_1,\ldots,M_k)\in sat(\Gamma)$, for  $f\in \Sigma_E$;\label{rule2}
\item if $C[S_1, \ldots,S_l]\stackrel{h}{\rightarrow}M$, where $C$ is an $E$-context such that $|C|\leq c_{E}$, and  $S_1, \ldots, S_l \in$ $sum_{\oplus}(sat(\Gamma))$, for some $AC$ symbol $\oplus$, then there exist an $ E$-context $C'$, a term $M'$, and terms $S_1', \ldots, S_k' \in sum_{\oplus}(sat(\Gamma))$, such that $|C'|\leq c_{E}^2$, and $M\stackrel{*}{\rightarrow}_{\mathcal{R}\cup AC}M'=_{AC}C'[S_1', \ldots, S_k']$;\label{rule3}
\item if $M\in sat(\Gamma)$ then $M\downarrow\in sat(\Gamma)$.\label{rule4}
\item if $M\in sat(\Gamma)$ then $\Gamma \vdash M$.\label{rule5}
\end{enumerate}
\end{definition}

Notice that the  set $sat(\Gamma)$ may not be unique. Any set $sat(\Gamma)$  satisfying the five conditions is adequate for the results.
\begin{remark}
The addition of rule 4 in the Definition \ref{locallystable} is necessary to prove case 1b of Lemma \ref{lemma:epcloseness}, where the rewriting reduction occurs in a term $M_i\in sat(\Gamma)$ in a position different from the ``head''.  Normal forms are strictly necessary in the set $sat(\Gamma)$, they are essential to lift the applications of rewriting rules in the head of ``small'' contexts to applications of rewriting rules in arbitrary positions of ``small'' contexts. With this  additional condition, Lemma 11 in \cite{AbadiCo2006} can also be proved. This fact was confirmed via personal communication with the second author of \cite{AbadiCo2006}. 
\end{remark}

The lemma and the corollary below, adapted from~\cite{AbadiCo2006}, are used in the proof of Theorem~\ref{theorem:epdecidability}. 

\begin{lemma}\label{lemma:epcloseness}
Let $E$ be a locally stable theory and $\Gamma=\{M_1,\ldots,M_n\}$ a set of  ground terms  in normal form. For every $E$-context $C_1$, for every $M_i \in sat(\Gamma)$, for every term $T$ such that $C_1[M_1,\ldots,M_k]\rightarrow_{\mathcal{R}\cup AC} T$, there exist an $E$-context $C_2$, and terms $M_i' \in sat(\Gamma)$, such that $T \stackrel{*}{\rightarrow}_{\mathcal{R}\cup AC} C_2[M_1', \ldots, M_l']$.
\end{lemma}

\begin{proof}

Suppose that $C_1[M_1,\ldots,M_k]\rightarrow_{AC}T$, for an $E$-context $C_1$ and $M_i\in\,sat(\Gamma)$. The proof is divided in two cases:
\begin{enumerate}
\item The reduction happens inside one of the terms $M_i$:
\begin{enumerate}
  \item if $M_i\stackrel{h}{\rightarrow}M_i'$ then by definition of $sat(\Gamma)$ (since $E$ is locally stable), there exist an $E$-context $C$ such that $|C|\leq c_E^2$ and $M_i'\stackrel{*}{\rightarrow}C[S_1,\ldots,S_l]$ where $S_j\in sum_{\oplus}(sat(\Gamma))$.
  
 Each $S_j\in sum_{\oplus}(sat(\Gamma))$  is of the form 
  $S_j=(\alpha_1\cdot_{\oplus}M_{j_1})\oplus \ldots\oplus (\alpha_n \cdot_{\oplus}M_{j_n}),$ 
  for $M_{j_k}\in sat(\Gamma)$. That is, $S_j=C_j[M_{j_1},\ldots, M_{j_k}]$, for $1\leq j\leq l$.
  Therefore,
  \begin{equation}
  \begin{split}
  C_1[M_1,\ldots, M_i, \ldots,M_k]\stackrel{h}{\rightarrow}C_1[M_1,\ldots, M'_i, \ldots,M_k]&\stackrel{*}{\rightarrow}_{AC}C_1[M_1,\ldots, C[S_1,\ldots, S_l], \ldots,M_k]\\
  &=_{AC}C_2[M_1^{''},\ldots, M_s^{''}],
  \end{split}
  \end{equation}
  where $M_t^{''} \in sat(\Gamma)$, for $1\leq t\leq s$.
   \item if $M_i\rightarrow_{AC}M_i'$ in a position different from ``head'', then \label{case:correction}
$$  C_1[M_1,\ldots, M_i, \ldots,M_k]\rightarrow C_1[M_1,\ldots, M'_i, \ldots,M_k]\stackrel{*}{\rightarrow}_{AC}C_1[M_1,\ldots, M_i\downarrow, \ldots,M_k].$$
 By case 4 in Definition \ref{locallystable}, $M_i\downarrow \in sat(\Gamma)$.
\end{enumerate}

\item The case where the reduction does not occur inside the terms $M_i$: this case if very technical and will be omitted here.  The complete proof can be found in the extended version of this paper.
\end{enumerate}

\end{proof}

As a consequence we obtain the following Corollary:

\begin{corollary}[\cite{AbadiCo2006}]\label{corollary:epcloseness}
Let $E$ be a locally stable theory. Let $\Gamma=\{M_1,\ldots,M_n\}$ be a set of ground terms  in normal form. For every $E$-context $C_1$, for every $M'_i\in sat(\Gamma)$, for every $T$ in normal form such that $C_1[M'_1,\ldots,M'_k]\stackrel{*}{\rightarrow}_{\mathcal{R}\cup AC}T$, there exist an $E$-context $C_2$ and terms $M_j{''}\in sat(\Gamma)$ such that $T=_{AC}C_2[M^{''}_1,\ldots,M_l^{''}]$.
\end{corollary}

\begin{proof}
The proof is the same as in \cite{AbadiCo2006}.
\end{proof}

In the following we show that any term $M$  deducible from $\Gamma$ is equal modulo AC to an $E$-context over terms in $sat(\Gamma)$.

\begin{lemma}[\cite{AbadiCo2006}]\label{lemma:deductioncontext}
Let $E$ be a locally stable theory. Let $\Gamma =\{M_1,\ldots,M_n\}$ be a finite set of ground terms in normal form, and $M$ be a ground term in normal form. Then $\Gamma \vdash M$ if and only if there exist an $E$-context $C$ and terms $M'_1, \ldots,M'_k\in sat(\Gamma)$ such that $M=_{AC}C[M^{'}_1,\ldots,M^{'}_n]$.
\end{lemma}
\begin{proof} 
The proof is the same as in \cite{AbadiCo2006}.
\end{proof}

As a consequence of the previous results decidability of IDP for locally stable theories is obtained:

\begin{theorem}\label{theo:IDPdecidable}
The Intruder Deduction Problem is decidable for locally stable theories.
\end{theorem}

In the next section we will provide a complexity bound for the decidability of the intruder deduction problem for a restricted case of locally stable theories.

\section{Locally Stable Theories with Inverses}

In order to obtain the polynomial complexity bound of our decidability algorithm we will need to consider the existence of inverses for each $AC$ symbol in the signature of our equational theory. Our algorithm will rely on solving systems of linear Diophantine equations over $\mathbb{Z}$ and the inverses will be interpreted as \emph{negative integers}.
 
(*) \emph{In the following results, let $E$ be a locally stable theory whose signature $\Sigma_E$ contains,  for each  $AC$ function symbol $\oplus$, its corresponding \emph{inverse} $i_{\oplus}$.}

That is, the following results are related to equational theories $E$ containing the following equation:
\begin{equation}
x\oplus i_{\oplus} (x) = e_{\oplus}
\end{equation}
for each AC-symbol $\oplus$ in $\Sigma_E$, where $i_{\oplus}$ is the unary function symbol representing the inverse of $\oplus$ and $e_{\oplus}$ is the corresponding neutral element.

\begin{definition}[Locally Stable with Inverses]\label{def:locallystableinverses}
An AC-convergent equational theory $E$  satisfying (*) is \emph{locally stable} if, for every finite set $\Gamma=\{M_1, \ldots,M_n\}$, where the terms $M_i$ are ground and in normal form, there exists a finite and computable set $sat(\Gamma)$, closed modulo $AC$, such that
\begin{enumerate}
\item $M_1, \ldots, M_n \in sat(\Gamma)$, $e_{\oplus} \in sat(\Gamma)$ for each $\oplus \in \Sigma_{E}$;
\item if $M_1,\ldots,M_k \in sat(\Gamma)$ and $f(M_1,\ldots,M_k)\in st(sat(\Gamma))$ then $f(M_1,\ldots,M_k)\in sat(\Gamma)$, for $f\in \Sigma_E$;
\item if $C[S_1, \ldots,S_l]\stackrel{h}{\rightarrow}M$, where $C$ is an $E$-context such that $|C|\leq c_{E}$, and  $S_1, \ldots, S_l \in sum_{\oplus}(sat(\Gamma))$, for some $AC$ symbol $\oplus$, then there exist an $ E$-context $C'$, a term $M'$, and terms $S_1', \ldots, S_k' \in sum_{\oplus}(sat(\Gamma))$, such that $|C'|\leq c_{E}^2$, and $M\stackrel{*}{\rightarrow}_{\mathcal{R}\cup AC}M'=_{AC}C'[S_1', \ldots, S_k']$;
\item if $M\in sat(\Gamma)$ then $M\downarrow\in sat(\Gamma)$.
\item if $M\in sat(\Gamma)$ then $i_{\oplus}(M)\downarrow\in sat(\Gamma)$ for each AC symbol $\oplus$ in $E$.
\item if $M\in sat(\Gamma)$ then $\Gamma \vdash M$.
\end{enumerate}
\end{definition}

Based on a well-founded ordering over the symbols in the language,  we prove that a restricted case of higher-order AC-matching (``is there an $E$-context $C$ such that $M=_{AC}C[M_1,\ldots,M_k]$ for some $M_1,\ldots,M_k\in sat(\Gamma)$?'')  can be solved in polynomial time in $|sat(\Gamma)|$ and $|M|$. This AC-matching problem is solved using the DO-ACM (Distinct-Occurrences of AC-matching)~\cite{narendran}, where every variable in the term being matched occurs only once. In addition, we also use a standard and polynomial time algorithm for solving SLDE over $\mathbb{Z}$~\cite{BouConDe,ClauFor,frumkin,papadimitriou}.

To facilitate the description of the algorithm below we have considered only one AC-symbol $\oplus$ whose corresponding inverse will be denoted by $i$. The proof can be extended similarly for theories with multiple AC-symbols each one with its corresponding inverse.

\begin{lemma}\label{lemma:acmatching}
Let $E$ be a locally stable theory satisfying (*), $\Gamma=\{M_1,\ldots,M_n\}$ a finite set of ground messages in normal form and $M$ a  ground term in normal form. Then the question of whether there exists an $E$-context $C$ and $T_1,\ldots,T_k\in sat(\Gamma)$  such that   $M=_{AC}C[T_1,\ldots,T_k]$ is decidable in polynomial time in $|M|$ and $|sat(\Gamma)|$.
\end{lemma}

\begin{proof}
Given $\Gamma$, we construct the set $sat(\Gamma)=\{T_1,\ldots, T_s\}$, which is computable and finite by Definition~\ref{locallystable}. We can then check whether $M=^?_{AC}C[T_1,\ldots,T_k]$ for  some $E$-context $C$ and terms $T_1,\ldots,T_k\in sat(\Gamma)$ using the following algorithm which is divided in its main component A), and procedures B) and C) for reducing linear Diophantine equations and selecting $T_i$'s from $sat(\Gamma)$, respectively.

\noindent A) \textbf{Algorithm 1.}

\begin{enumerate}
\item For all positions $p$ in $M$ headed by $\oplus$ starting from the longest positions in decreasing order (positions seen as sequences) solve the \emph{system of linear Diophantine equations} (see part B below) for $M|_p$ with $sat(\Gamma)\cup S$, where $S$ is built incrementally from $sat(\Gamma)$, starting with $S_0=\emptyset$, including all $M|_p$ that have solutions. In other words: 

Let $\mathcal{P}'=\{p_1,\ldots,p_t\}$ be the set of positions of $M$ such that $M|_p$ is headed with $\oplus$, organised in decreasing order. For each $p_j \in \mathcal{P}'$ let $M|_{p_j}$ be the subterm of $M$ such that 
$$M|_{p_j}=n_{j_1}\oplus \ldots \oplus n_{j_{kj}}\,\, (j=1,\ldots, t)$$ 

Recursively  find, but suppressing step 1 in this recursive call, solutions for the arguments $n_{j_{i_1}},\ldots, n_{j_{i_l}}$ of $M|_{p_j}$ with $n_{j_{im}} \in \{n_{j_1},\ldots,n_{j_{k_j}} \}$ with respective $E$-contexts $C_{j_{i_1}},\ldots, C_{j_{i_l}}$ such that 
$$n_{j_{i_m}}=C_{j_{i_m}}[T_1,\ldots,T_{s_{i_m}}]$$
where $T_q \in sat(\Gamma)\cup S_{j-1}$, $q = 1,\ldots, s_{i_m}$.

Then one checks satisfiability of the SLDE generated from $M|_{p_j}$ and $sat(\Gamma)\cup S_{j-1} \cup \{n_{j_{i_1}},\ldots, n_{j_{k_l}}\}$ (see steps B and C).

If there is a solution then $S_j:= S_{j-1} \cup \{n_{j_{i_1}},\ldots, n_{j_{k_l}}\}\cup \{M|_{p_j}\}$

\item Let $S:= S_t$. Classify the terms in $sat(\Gamma)\cup S$ by size.
\item For each term $T_i \in sat(\Gamma) \cup S$ (from terms of maximal size to terms of minimal size) check:

    \begin{itemize}
      \item For each position $q\in \mathcal{P}os(M)$ such that $T_i=_{AC}M|_q$ do      
             
              Check whether the path between $T_i$ and the root of $M$ contains a $\oplus$:
               \begin{itemize}
                \item if NOT, then delete $M|_q$ from $M $ and move to $T_{i+1}$.
                 \item if YES (there is a $\oplus$) then $M$ has a subterm  $ N$ such that $N=n_1\oplus \ldots \oplus n_j[T_i]\oplus \ldots \oplus n_k$ and $N$ cannot be constructed from $sat(\Gamma)\cup S$. Therefore,  $M$ cannot be written as an $E$-context with terms from $sat(\Gamma)$.
           
    \end{itemize}
    \end{itemize}

\item Check whether the remaining part of $M$ still contains  $E$-aliens. If it is not  the case, we have found an $E$-context $C$ and terms $M_1,\ldots,M_k \in sat(\Gamma)$ and $M=_{AC}C[M_1,\ldots,M_k]$; otherwise such an $E$-context does not exist.
\end{enumerate}

\noindent B) \textbf{Reduction to linear Diophantine equations.} 

First, notice that, for each position $p$ such that $M|_p$ is headed with $\oplus$ we have
\begin{equation}\label{Eq:eq1}
M|_p=\alpha_1 m_1\oplus \ldots \oplus \alpha_r m_r\, , \, \alpha_j \in \mathbb{N} 
\end{equation}
where $m_j$ is not headed with $\oplus$ and $\alpha_jm_j$ counts for $\underbrace{m_j\oplus \ldots \oplus m_j}_{\alpha_j-times}$.

 We want to prove that there are $\beta_1,\ldots,\beta_q \in \mathbb{N}$ such that 
\begin{equation}\label{Eq:eq2}
\beta_1T_1\oplus \ldots \oplus \beta_qT_q =_{AC} M|_p=\alpha_1 m_1\oplus \ldots \oplus \alpha_r m_r
\end{equation}
This AC-equality is only possible when
$T_i=\gamma_{1i}m_1\oplus \ldots\oplus \gamma_{ri}m_r $
for each $i$, $1\leq i \leq q\leq s$ and $\gamma_{j_i}\in \mathbb{N}$.

That is, $\beta_1T_1\oplus \ldots \oplus \beta_qT_q=_{AC}\alpha_1 m_1\oplus \ldots \oplus \alpha_r m_r$
if and only if
\begin{equation}
\begin{split}
&\beta_1(\gamma_{1_1}m_1\oplus \ldots\oplus \gamma_{r_1}m_r)\oplus \beta_2(\gamma_{1_2}m_1\oplus \ldots\oplus \gamma_{r_2}m_r)\oplus\ldots \\
&\ldots\oplus \beta_q(\gamma_{1_q}m_1\oplus \ldots\oplus \gamma_{r_q}m_r)=\alpha_1 m_1\oplus \ldots \oplus \alpha_r m_r
\end{split}
\end{equation}

 if and only if
\begin{equation}
\begin{split}
&(\gamma_{1_1}\beta_1 \oplus \gamma_{1_2}\beta_2 \ldots\oplus \gamma_{1_q}\beta_q)m_1\oplus (\gamma_{2_1}\beta_1 \oplus \gamma_{2_2}\beta_2 \ldots\oplus \gamma_{2_q}\beta_q)m_2\oplus \ldots\\
&\ldots(\gamma_{r_1}\beta_1 \oplus \gamma_{r_2}\beta_2 \ldots\oplus \gamma_{r_q}\beta_q)m_r=\alpha_1 m_1\oplus \ldots \oplus \alpha_r m_r
\end{split}
\end{equation}
if and only if

\begin{equation}
S=\left\{
\begin{split}
 \gamma_{1_1}\beta_1 \oplus \gamma_{1_2}\beta_2 \ldots\oplus \gamma_{1_q}\beta_q&=\alpha_1\\
  \gamma_{2_1}\beta_1 \oplus \gamma_{2_2}\beta_2 \ldots\oplus \gamma_{2_q}\beta_q&=\alpha_2\\
  \vdots \hspace{1cm} &\\
 \gamma_{r_1}\beta_1 \oplus \gamma_{r_2}\beta_2 \ldots\oplus \gamma_{r_q}\beta_q&=\alpha_r\\
\end{split}
\right.
\end{equation}
where $S$ is a system of linear Diophantine equations over $\mathbb{Z}$  which can be solved in polynomial time~\cite{BouConDe,ClauFor,frumkin,papadimitriou}.

\begin{remark}
We will interpret the equations \ref{Eq:eq1} and \ref{Eq:eq2} inside integer arithmetic. If there exists an index $j$ such that $m_j= i(m_j')$  and $m_j' $ is not headed with $i$ then $\alpha_j m_j=\alpha_j (i(m_j'))$ and we will take it as $(-\alpha_j)m_j'$. Therefore, we can take $\alpha_j \in \mathbb{Z}$, for all $j$. We can use the same reasoning to conclude that $\beta_j \in \mathbb{Z}$, for all $1\leq j\leq q$ and $\gamma_{j_i} \in \mathbb{Z}$, for all $i$ and $j$.
\end{remark}

\noindent C) \textbf{Selecting the $T_j's$ from $sat(\Gamma)$.}

  For each $T_i \in sat(\Gamma)$, $1\leq i\leq s$ we want to check if $T_i=\gamma_{1_i}m_1\oplus \ldots\oplus \gamma_{r_i}m_r $.

\textbf{Algorithm 2:}

For each $T_i \in sat(\Gamma)$, $1\leq i \leq s$, solve the equation
$T_i \oplus x_i =_{AC} \alpha_{1}m_1 \oplus \ldots\oplus \alpha_{r}m_r $
where $x_i$ is a fresh variable.

Since the $T_i's$ and $M$ are ground terms, this equation can be seen as an instance of the DO-ACM matching problem which can be solved in time $\mathcal{O}(|T_i\oplus x_i|.|M|_p|)$~\cite{narendran}.

If there exists $T_i \in sat(\Gamma)$ such that $T_i=\gamma_{1_i}^*m_1\oplus \ldots\oplus \gamma_{r_i}^*m_r\oplus u $, where $u$ is not empty, $\gamma_{i_j}^*\in \mathbb{N} $ and the \textbf{Algorithm 2} can no longer be applied then $T_i$ will not be selected.

Notice that each step of the algorithm can be done in polynomial time in $|M|$ and $|sat(\Gamma)|$. Therefore, the whole procedure is polynomial in $|M|$ and $sat(\Gamma)$.
\end{proof}

\begin{remark}
 For the proof we can adopt an ordering in which, for instance, variables are smaller than constants, constants smaller than function symbols, and function symbols are also ordered, but other suitable order can be used.  Terms are compared by the associated lexicographical ordering  built from this ordering on symbols. 
\end{remark}

\begin{example}[Finite Abelian Groups]
We consider the theory of Abelian Groups where the signature is $\Sigma_{AG}=\{+,0,i\}$ for $i$ the inverse function and $+$ the AC group operator. The equational theory $E_{AG}$ is:
$$
E_{AG}=\left\{
\begin{array}{l@{\hspace{1cm}}c @{\hspace{1cm}}r}
\begin{array}{rcl}
x+(y+z)&=&(x+y)+z\\
x+y&=&y+x\\
i(x+y)&=&i(y)+i(x)\\
\end{array}
&
\begin{array}{rcl}
x+0&=&x\\
x+i(x)&=&0
\end{array}
&
\begin{array}{rcl}
i(i(x))&=&x\\
i(0)&=&0
\end{array}
\end{array}
\right.
$$
We define $\mathcal{R}_{AG} $ by orienting the equations from left to right (excluding the equations for associativity and commutativity). $\mathcal{R}_{AG}$ is AC-convergent. The size $c_{E_{AG}}$ of the theory is at least 5. In the following prove that  $E_{AG}$ is locally stable with inverses for finite models, i.e., we define a set $sat(\Gamma)$ satisfying the properties in the Definition \ref{locallystable}.  For a given set $\Gamma=\{M_1,\ldots,M_k\}$ of ground terms in normal form, $sat(\Gamma)$ is the smallest set such that:
\begin{enumerate}
\item $M_1,\ldots,M_k\in sat(\Gamma)$;
\item $M_1,\ldots,M_k\in sat(\Gamma)$ and $f(M_1,\ldots,M_k)\in st(sat(\Gamma))$ then $f(M_1,\ldots,M_k)\in sat(\Gamma)$, $f\in \Sigma_{AG}$;
\item if $M_i,M_j \in sat(\Gamma)$ and $M_i+M_j\stackrel{h}{\rightarrow}M$ via rule $x+i(x)\rightarrow 0$ then $M\downarrow\in sat(\Gamma)$;
\item if $M_j \in sat(\Gamma)$ then $i(M_j)\downarrow \in sat(\Gamma)$;
\item if $M_i=_{AC}M_j$ and $M_i \in sat(\Gamma)$ then $M_j \in sat(\Gamma).$
 \end{enumerate}  

The set $sat(\Gamma)$ defined for Finite Abelian Groups is finite. 

\end{example}

Although it was said in \cite{AbadiCo2006} that the theory of Abelian Groups is locally stable, no proof of such fact was found in the literature. With the proviso that the Abelian Group under consideration is finite, we have demonstrated that $|sat(\Gamma)|$ is exponential in the size of $|\Gamma|$.


These results give rise to the decidability of deduction for locally stable theories. Notice that  polynomiality on $|sat(\Gamma)|$ relies on the use of the AC-matching algorithm proposed in Lemma~\ref{lemma:acmatching}. Unlike~\cite{AbadiCo2006},  we do not need to compute of the congruence class modulo AC of $M$ (which may be exponential). This gives us a slightly different version of the decidability theorem:

\begin{theorem}\label{theorem:epdecidability}
Let $E$ be a locally stable theory satisfying (*). If $\Gamma=\{M_1,\ldots,M_n\}$ is a finite set of ground terms in normal form and $M$ is a ground term in normal form, then $\Gamma\vdash M$ is decidable in polynomial time in $|M|$ and $|sat(\Gamma)|$.
\end{theorem}

\begin{proof}
The result follows directly from Lemmas \ref{lemma:acmatching} and \ref{lemma:deductioncontext}.
\end{proof}

In the following example we consider the  \emph{Pure AC-theory} which can be proven to be locally stable but does not contain the inverse of the AC-symbol $+$.

\begin{example}[Pure $AC$ Theory]
$\Sigma_{AC}$ contains only constant symbols, the AC-symbol $\oplus$ and the equational theory contains only the $AC$ equations for $\oplus$:
$$
AC=\left\{
\begin{array}{l@{\hspace{3cm}}r}
x\oplus y= y\oplus x
&
x\oplus(y\oplus z) =(x\oplus y)\oplus z
\end{array}
\right\}
$$
In this case, $E=AC$  and  $ \mathcal{R}=\emptyset$ is the  AC-convergent TRS associated to $E$. Let  $\Gamma=\{M_1,\ldots,M_k\}$ be a finite set of ground terms in normal form. Let us define $sat(\Gamma)$ for the pure $AC$ theory as the smallest set such that 
\begin{enumerate}
\item $M_1,\ldots, M_k\in sat(\Gamma)$;
\item if $M_i,M_j\in sat(\Gamma)$ and $M_i\oplus M_j \in st(sat(\Gamma))$ then $M_i\oplus M_j\in sat(\Gamma)$.
\item if $M_i=_{AC}M_j$ and $M_i\in sat(\Gamma)$ then $M_j\in sat(\Gamma)$.
\end{enumerate}
The set $sat(\Gamma)$ is finite since we add only terms whose size is smaller or equal than the maximal size of the terms in $\Gamma$. It is easy to see that the set $sat(\Gamma)$ satisfies the rules  \ref{rule1},\ref{rule2}, \ref{rule4} and \ref{rule5}. Since $\mathcal{R}=\emptyset$ it follows that \ref{rule3} is also satisfied. Therefore,
 $AC$ is locally stable.

\emph{ The size  of $sat(\Gamma)$:} 
\begin{itemize}
\item Steps 1 and 2: only subterms in $sat(\Gamma)$ are added.
\item Step 3: for each $M_i\in sat(\Gamma)$  add $M_j=_{AC}M_i\in sat(\Gamma)$. Notice that the number of terms added in $sat(\Gamma)$, in this case, depends on the number of occurrences of $\oplus$ in $M_i$. Suppose that $M_i$ contains $n$ occurrences of $\oplus$: 
$$M_i=M_{i_1}\oplus \ldots \oplus M_{i_{n+1}}.$$
There are $(n+1)!$ terms $M_j$ such that $M_1=_{AC}M_j$.
\end{itemize}
Suppose that each $M_i$ in $\Gamma$ contains $n_i$ occurrences of $\oplus$.
Then, $|M_i|=\displaystyle\sum_{j=1}^{n_i+1}|M_{i_j}|+n_i.$ Let $n=\max_{1\leq i\leq k}\{n_i\}$. There exists an index $r$ such that $M_r$ contains $n_r=n$ occurrences of $\oplus$.
Since $|\Gamma|=\displaystyle\sum_{i=1}^k|M_i|$ it follows that 
$ n \leq |M_r|-\displaystyle\sum_{j=1}^{n+1}|M_{r_j}|\leq |\Gamma|.$
Then the number of terms added in step 3 is  
$\displaystyle\sum_{i=1}^k (n_i+1)!\leq (n+1)! \cdot k \leq (|\Gamma|+1)!\cdot k .$

\end{example}

\begin{remark} In this case one can adapt Lemma \ref{lemma:acmatching} such that the algorithm would rely on solving systems of linear Diophantine equations over $\mathbb{N}$ which is NP-complete~\cite{papadimitriou}. Therefore, the complexity of IDP for pure AC would be exponential, agreeing with previous results~\cite{lafourcade}.
\end{remark}

\section{Elementary Deduction Problem for Locally Stable Theories}
To establish necessary concepts for the next results, we recall the well-known translation between  natural deduction and sequent calculus systems to model the IDP as a proof search in sequent calculus, whose properties (such as cut or subformula) facilitate the study of decidability of deductive systems. For an AC-convergent equational theory E, the System $\mathcal{N}$ in Table \ref{equationalreasoning} is equivalent to the $(id)$-rule of the sequent calculus (Table \ref{DeductionRulesForIntruder}) introduced in~\cite{TiGo2009}:
\begin{prooftree}
\AxiomC{$\stackrel{M\approx_{E} C[M_1,\ldots, M_k]}{\textnormal{C[\;] an E-context, and } M_1,\ldots,M_k\in \Gamma} $}
\RightLabel{($id$)}
\UnaryInfC{$\Gamma \vdash M$}
\end{prooftree}

 Consequently,  IDP for System $\mathcal{N}$ is  equivalent to the \emph{Elementary Deduction Problem}: 
\begin{definition}
Given an AC-convergent equational theory $E$ and a sequent $\Gamma \vdash M$ ground and in normal form, the \emph{elementary deduction problem} (EDP) for $E$, written $\Gamma \Vdash_{E}M$, is the problem of deciding whether the $(id)$-rule is applicable in $\Gamma\vdash M$. 
\end{definition}
The theorem below decides EDP  for  locally stable theories :

\begin{theorem}\label{theorem:edpisptime}
Let $E$ be a locally stable equational theory satisfying (*). Let $\Gamma \vdash M$ be a  ground sequent  in normal form. The \emph{elementary deduction problem} for the theory $E$ ($\Gamma \Vdash_E M$) is decidable in polynomial time in $|sat(\Gamma)|$ and $|M|$.
\end{theorem}

\begin{proof}
By Lemma \ref{lemma:acmatching}, the problem whether $M=_{AC}C[M_1,\ldots,M_k]$ for an $E$-context $C$ and terms $M_1,\ldots,M_k\in sat(\Gamma)$ is decidable in polynomial time in $|sat(\Gamma)|$ and $|M|$.
If $M=_{AC}C[M_1,\ldots,M_k]$ for an $E$-context $C$ and terms $M_1,\ldots,M_k\in sat(\Gamma)$ then  there exist an $ E$-context $C'$ and terms $M'_1,\ldots,M'_n\in \Gamma$ such that
$C[M'_1,\ldots,M'_n]\stackrel{*}{\rightarrow}_{\mathcal{R}\cup AC} M.$
It is enough to observe that for all $T\in sat(\Gamma)$, $T$ can be constructed from the terms in $\Gamma$.

If there is no $E$-context $C$ and terms $M_1,\ldots,M_k\in sat(\Gamma)$ such that $M=_{AC}C[M_1,\ldots,M_k]$ then, by Corollary \ref{corollary:epcloseness}, there are no $\textsf E$-context and terms  $M'_1, \ldots, M'_t\in sat(\Gamma)$ such that $C[M'_1, \ldots, M'_t]\stackrel{*}{\rightarrow}_{\mathcal{R}\cup AC} M.$ Therefore,  there is no $E$-context $C''$ and terms $M''_1,\ldots,M''_l\in \Gamma$ such that $C''[M''_1, \ldots, M''_l]\stackrel{*}{\rightarrow}_{\mathcal{R}\cup AC} M.$
Thus, the EDP for $E$ is decidable in polynomial time in $|sat(\Gamma)|$ and $|M|$.
\end{proof}

\subsection{Extension with Blind Signatures}\label{extension}
Blind signature is a basic cryptographic primitive in e-cash. This concept was introduced by David Chaum in \cite{Chaum} to allow a  bank (or anyone) sign messages without seeing them. David Chaum's idea  was to use this homomorphic property in such a way that Alice can multiply the original message with a random (encrypted) factor that will make the resulting image meaningless to the Bank. If the Bank agrees to sign this random-looking data and return it to Alice, she is able to divide out the blinding factor such that the Bank's signature in the original message will appear.

Given a locally stable equational theory $E$, we extend the signature $\Sigma_{E}$ with $\Sigma_C$, a set containing function symbols for ``constructors'' for blind signatures, in order to obtain decidability results for the extension of the IDP for System $\mathcal{N}$ taking into account some rules for blind signatures. 

\subsubsection*{Extended Syntax}

The signature  $\Sigma$ consists of function symbols and is defined by the union of two sets: $\Sigma =\Sigma_C \cup \Sigma_{E}$ ( with $\Sigma_{\textsf{E}}\cap \Sigma_C=\emptyset$), where 
$$\Sigma_C=\left\{\textsf{pub}(\_), \textsf{sign}(\_\; ,\_), \textsf{blind}(\_\; ,\_), \left\{\_\right\}_{\_}, <\_\; ,\_>\right\}$$ represents the \emph{constructors}, whose interpretations are: $\textsf{pub}(M)$ gives  the public key generated from a private key $M$; $\textsf{blind}(M,N)$ gives $M$  encrypted with $N$ using blinding encryption; $\textsf{sign}(M,N)$ gives $M$ signed with a private key  $N$; $\left\{M\right\}_N$ gives $M$ encrypted with the key $N$ using Dolev-Yao symmetric encryption; $\langle M,N\rangle$ constructs a pair of terms from $M$ and $N$.
Then the extended grammar of the set of \emph{terms} or messages is given as
\begin{equation*}
M,N \;:= a \;|\; x \;| f(M_1,\ldots, M_n)|\textsf{pub}(M) | \textsf{sign}(M,N) | \textsf{blind}(M,N)|\left\{M\right\}_N|\langle M, N\rangle
\end{equation*}

Notice that, with the extension  an $E$-alien term $M$ is a term headed with $f\in \Sigma_C$ or $M$ is a private name/constant. An $E$-alien subterm $M$ of $N$ is said to be an $E$-\emph{factor} of $N$ if there is another subterm $F$ of $N$ such that $M$ is an immediate subterm of $F$ and $F$ is headed by a symbol $f\in\Sigma_{E}$. This notion can be extended to sets in the obvious way:  a term $M$ is an $E$-factor of  $\Gamma$ if it is an $E$-factor of a term in $\Gamma$. These notions  were introduced in \cite{TiGo2009}.

The operational meaning of each constructor will be defined by their corresponding inference rules in the sequent calculus to be described.

\subsubsection*{Extending the EDP to Model Blind Signatures} 

Following the approach proposed in \cite{TiGo2009}, we extend EDP with blind signatures using the sequent calculus $\mathcal{S}$ described in Table \ref{DeductionRulesForIntruder}. In this way, we can model intruder deduction for the combination of a locally stable theory $E$ with blind signatures in a modular way: the theory $E$ is used in the $id$ rule, while blind signatures are modelled with additional deduction rules. As shown below, this approach has the advantage that we can derive decidability results for the intruder deduction problem without needing to prove that the combined theory is locally stable (in contrast with the results in the previous section and in~\cite{AbadiCo2006}).

\begin{table}[ht]
\caption{System $\mathcal{S}$\;: Sequent Calculus for the Intruder}
\label{DeductionRulesForIntruder}
{\small
\hrule
\begin{prooftree}
\AxiomC{$\stackrel{M\approx_{E} C[M_1,\ldots, M_k]}{\textnormal{C[\;] an E-context,} M_1,\ldots,M_k\in \Gamma} $}
\RightLabel{($id$)}
\UnaryInfC{$\Gamma \vdash M$}
\DisplayProof
\defaultHypSeparation 
\hspace{8mm}
\AxiomC{$\Gamma \vdash M$}
\AxiomC{$\Gamma, M \vdash T$}
\RightLabel{($cut$)}
\BinaryInfC{$\Gamma \vdash T$}
\end{prooftree}

\begin{prooftree}
\AxiomC{$\Gamma, \left\langle M,N \right\rangle , M,N \vdash T$}
\RightLabel{($p_L$)}
\UnaryInfC{$\Gamma,\left\langle M,N \right\rangle \vdash T$}
\DisplayProof
\defaultHypSeparation 
\hspace{3cm}
\AxiomC{$\Gamma \vdash M$}
\AxiomC{$\Gamma \vdash N$}
\RightLabel{($p_R$)}
\BinaryInfC{$\Gamma \vdash \left\langle M,N \right\rangle$}
\end{prooftree}

\begin{prooftree}
\AxiomC{$\Gamma \vdash M$}
\AxiomC {$\Gamma \vdash K$}
\RightLabel{($e_R$)}
\BinaryInfC{$\Gamma \vdash \left\{ M\right\}_K $}
\DisplayProof
\defaultHypSeparation 
\hspace{5mm}
\AxiomC{$\Gamma, \left\{ M\right\}_K \vdash K$}
\AxiomC{$\Gamma, \left\{ M\right\}_K ,M,K\vdash N$}
\RightLabel{($e_L$)}
\BinaryInfC{$\Gamma,\left\{ M\right\}_K \vdash N$}
\end{prooftree}

\begin{prooftree}
\AxiomC{$\Gamma \vdash M$}
\AxiomC{$\Gamma \vdash K$}
\RightLabel{\small{(\textsf{sign}$_R$)}}
\BinaryInfC{$\Gamma \vdash $ \textsf{sign}$(M,K)$}
\DisplayProof
\defaultHypSeparation 
\hspace{2.5cm}
\AxiomC{$\Gamma \vdash M$}
\AxiomC{$\Gamma \vdash K$}
\RightLabel{ \footnotesize{(\textsf{blind}$_R$)}}
\BinaryInfC{$\Gamma \vdash$ \textsf{blind}$(M,K)$ }
\end{prooftree}

\begin{prooftree}
\AxiomC{$\Gamma$,\textsf{sign}$(M,K)$, \textsf{pub}$(L), M \vdash N$}
\RightLabel{\footnotesize{$(\textsf{sign}_L) K=\!_{AC}\! L$}\hspace{4.5cm}}
\UnaryInfC{$\Gamma$,\textsf{sign}$(M,K)$, \textsf{pub}$(L)\vdash N$}
\end{prooftree}

\begin{prooftree}
\AxiomC{$\Gamma$,\textsf{blind}$(M,K)\vdash K$}
\AxiomC{$\Gamma$,\textsf{blind}$(M,K), M, K\vdash N$}
\RightLabel{ (\textsf{blind}$_{L_1}$)\hspace{3cm}}
\BinaryInfC{$\Gamma$,\textsf{blind}$(M,K)\vdash N$}
\end{prooftree}

\begin{prooftree}
\AxiomC{$\Gamma,\sign(\blind(M,R),K)\vdash R$}
\AxiomC{\hspace{-3mm}$\Gamma, \sign(\blind(M,R),K), \sign(M,K),R\vdash N$}
\RightLabel{ (\textsf{blind}$_{L_2}$)}
\BinaryInfC{$\Gamma,\sign(\blind(M,R),K)\vdash N$}
\end{prooftree}

\begin{prooftree}
\AxiomC{$\Gamma \vdash A$}
\AxiomC{$\Gamma, A \vdash M$}
\RightLabel{$(acut)$, $A$ is an $E$-factor of $\Gamma \cup \left\{M\right\}$\hspace{3cm}}
\BinaryInfC{$\Gamma \vdash M$}
\end{prooftree}
\hrule
}
\end{table}

Analysing the system $\mathcal{S}$ one can make the following observations:
\begin{enumerate}
\item The rules $p_L, e_L$,$ \textsf{sign}_L$,$ \textsf{blind}_{L1}$, $\textsf{blind}_{L2}$ and $acut$ are called \emph{left rules} with $\langle M,N\rangle$, $\{M\}_K$, $\textsf{sign}(M,K)$, $\textsf{blind}(M,K)$, $\textsf{sign}(\textsf{blind}((M,R),K)$ and $A$ as \emph{principal term}, respectively. The rules $p_R, e_R,\textsf{sign}_R$ and $ \textsf{blind}_{R}$ are called \emph{right rules}.
\item  The rule $(acut)$, called \emph{analytic cut} is necessary  to prove  cut rule \emph{admissibility}. A complete proof can be found in~\cite{TiGo2009,nantesayala}.
\end{enumerate}

\begin{remark} Considerations about locally stable theories with blind signatures:
\begin{enumerate}
\item All the results proved on Section \ref{sec:locallystable} are valid under this extension with blind signatures since the results depend only on the equational theory $E$ and on the  symbols in $\Sigma_E$. Unlike example 5.2.4~\cite{AbadiCo2006}, the theory of Blind Signatures is not considered as part of the equational theory, the functions are abstracted in the set of constructors with the operational meaning represented in the sequent calculus.
\item In \cite{TiGo2009} it is shown that the intruder deduction problem for $\mathcal{S}$ is \emph{polynomially reducible} to the EDP for $E$: \emph{if the EDP problem in $E$ has complexity $f(m)$ then the deduction problem $\Gamma\vdash M$ in $\mathcal{S}$ has complexity $O(n^k.f(n))$ for some constant $k$}\footnote{Here, $m$ is the size of the input of EDP and $n$ is the cardinality of the set $St(\Gamma\cup \{M\})$ defined in \cite{TiGo2009}}.  
This result was proved for an  AC-convergent equational theory $E$ containing only one $AC$ symbol and extended to finite a combination of disjoint AC-convergent equational theories each one containing only one AC-symbol. 
\item In \cite{nantesayala}, it was proved that deduction in $\mathcal{S}$ reduces polynomially to $EDP$ in the case of the AC-convergent equational theory $\ep$, which contains three different AC-symbols and rules for exponentiation and cannot be split into disjoint parts.
\end{enumerate}
\end{remark}

As a consequence of the results mentioned in the above remark, we can state the following result: 
\begin{corollary}\label{theo:eppblind}
Let $E$ be a locally stable theory satisfying (*) containing only one AC-symbol or formed by a finite and disjoint combination of AC-symbols.
Let $\Gamma$ a finite set of ground terms in normal form and $M$ a ground term  in normal form. The IDP for  the theory $E$ combined with blind signatures ($\Gamma\vdash M$) is decidable in polynomial time in $|sat(\Gamma)|$ and $|M|$.
\end{corollary}

\section{Conclusion}
We have shown that the IDP is decidable for locally stable
theories. In order to obtain the polynomiality result, a restriction
on the equational theory is necessary: the theory must contain
inverses of all AC-symbols.  We have proposed an algorithm to solve a restricted case of
higher-order AC-matching by using the DO-ACM matching algorithm
combined with an algorithm to solve linear Diophantine equations over
$\mathbb{Z}$. Based on this algorithm, we obtain a polynomial
decidability result for IDP for a class of locally stable theories
with inverses. Our algorithm does not need to compute the set of
normal forms modulo AC of a given term (which may be exponential). 
Therefore, we can conclude that the deducibility relation is decidable
in polynomial time for a very restricted class of equational theories,
it does not work for all locally stable theories as \cite{AbadiCo2006}
has claimed.  It also decides the IDP for the combination of locally
stable theories with the theory of blind signatures, using a
translation between natural deduction and sequent calculus.

\nocite{*}

\bibliographystyle{eptcs}

\begin{thebibliography}{1}
\providecommand{\bibitemdeclare}[2]{}
\providecommand{\surnamestart}{}
\providecommand{\surnameend}{}
\providecommand{\urlprefix}{Available at }
\providecommand{\url}[1]{\texttt{#1}}
\providecommand{\href}[2]{\texttt{#2}}
\providecommand{\urlalt}[2]{\href{#1}{#2}}
\providecommand{\doi}[1]{doi:\urlalt{http://dx.doi.org/#1}{#1}}
\providecommand{\bibinfo}[2]{#2}



\bibitem{AbadiCo2006}
M. Abadi and V. Cortier.
\newblock Deciding knowledge in security protocols under equational theories.
\newblock {\em Theoretical Computer Science}, 367(1-2):2--32, 2006.
\doi{10.1016/j.tcs.2006.08.032}.

\bibitem{appliedpicalculus}
M. Abadi and C. Fournet.
\newblock Mobile Values, New Names, and Secure Communication.
\newblock In {\em Proc.  28$^{th}$ ACM SIGPLAN-SIGACT symposium on Principles of programming languages (POPL'01)}, pages 104--115, 2001.
\doi{10.1145/360204.360213}.

\bibitem{spi}
M. Abadi and A.D.~Gordon
\newblock A Calculus for Cryptographic Protocols: The spi Calculus.
\newblock  {\em Information and Computation }, 148(1): 1--70, 1999.
\doi{10.1006/inco.1998.2740}.

\bibitem{avispa}
A. Armando \emph{et al}.
\newblock The AVISPA Tool for the Automated Validation of Internet Security Protocols and Applications.
\newblock In {\em Proc. 17$^{th}$ Computer Aided Verification (CAV'05)}, volume 3576, pages 281--285. Springer-Verlag 2005.
\doi{10.1007/11513988\_27}.

\bibitem{AFSextended}
M.~Ayala-Rinc\'{o}n, M.~Fern\'{a}ndez and D.~Nantes-Sobrinho.
\newblock Elementary Deduction Problems for Locally Stable Theories with Normal Forms (extended version).
\newblock \url{http://www.mat.unb.br/~dnantes/Publications}.

\bibitem{baader}
F.~Baader and T.~Nipkow.
\newblock{\em Term Rewriting and All That}.
\newblock CUP, 1998.

\bibitem{YAPA}
M.~Baudet, V.~Cortier and S.~Delaune.
\newblock YAPA: A Generic Tool for Computing Intruder Knowledge.
\newblock In {\em Proc. of  20$^{th}$ International
               Conference on Rewriting Techniques and Applications (RTA'09)}, volume 5595 of {\em LNCS}, pages 148-163. Springer, 2009.
\newblock{\url{arXiv:1005.0737}}, \doi{10.1007/978-3-642-02348-4\_11}.

\bibitem{narendran}
D. ~Benanav, D.~Kapur, P.~Narendran, and L.~Wang.
\newblock Complexity of matching problems.
\newblock In {\em Journal of Symbolic Computation}, 3(1/2): 203--216, 1987.
\doi{10.1007/3-540-15976-2\_22}.

\bibitem{BeC06}
V.~Bernat and H.~Comon-Lundh.
\newblock Normal proofs in intruder theories.
\newblock In {\em Proc. 11$^{th}$ Asian Computing Science Conference, Advances in Computer Science - Secure Software
               and Related Issues (ASIAN'06)}, volume 4435 of {\em LNCS}, pages 151--166. Springer-Verlag, 2006.
\doi{10.1007/978-3-540-77505-8\_12}.

\bibitem{proverif}
B.~Blanchet.
\newblock An Efficient Cryptographic Protocol Verifier Based on Prolog Rules.
\newblock In {\em Proc. 14$^{th}$ IEEE Computer Security Foundations Workshop (CSFW'01)}, pages 82--96, IEEE Comp. Soc., 2001.
\newblock{\url{http://doi.ieeecomputersociety.org/10.1109/CSFW.2001.930138}}.

\bibitem{cryptoverif}
B.~Blanchet.
\newblock A Computationally Sound Mechanized Prover for Security Protocols.
\newblock In {\em IEEE Transactions on Dependable and Secure Computing}, volume 5 (4), pages 193--207, 2008.
\doi{10.1109/TDSC.2007.1005}

\bibitem{BouConDe}
A.~Boudet, E.~Contejean and H.~Devie.
\newblock A new AC Unification Algorithm with an Algorithm for Solving Systems of Linear Diophantine Equations.
\newblock In {\em  Proc. 5$^{th}$ Annual Symposium on Logic in Computer
               Science (LICS '90)}, pages 289--299, 1990.
\doi{10.1109/LICS.1990.113755}.

\bibitem{Bursucconstraints}
B.~Bursuc, H.~Comon-Lundh, and S.~Delaune.
\newblock Deducibility constraints, equational theory and electronic money.
\newblock In {\em Rewriting, Computation and Proof, Essays Dedicated to Jean-Pierre
               Jouannaud on the occasion of his 60th Birthday}, volume 4600 of {\em
 LNCS}, pages 196--212. Springer-Verlag, 2007.
\doi{10.1007/978-3-540-73147-4\_10}.

\bibitem{Chaum}
D. ~Chaum.
\newblock Blind Signatures for Untraceable Payments.
\newblock In {\em Proc. of Advances in Cryptology  (CRYPTO'82)}, pages 199--203, Plenum Press, 1982. 
\newblock{\url{http://blog.koehntopp.de/uploads/Chaum.BlindSigForPayment.1982.PDF}}.

\bibitem{ClauFor}
M.~Clausen and A.~Fortenbacher.
\newblock Efficient Solution of Linear Diophantine Equations.
\newblock In {\em Journal of Symbolic Computation}, Volume 8(1-2), pages 201--216, 1989.
\doi{10.1016/S0747-7171(89)80025-2}.

\bibitem{CoLuSh03}
H.~Comon-Lundh and V. Shmatikov. 
\newblock Intruder Deduction, Constraint Solving and Insecurity Decisions in Presence of Exclusive or.
\newblock In {\em Proc. 18$^{th}$ IEEE Symposium on Logic in Computer Science (LICS'03)}, pages 271--280. IEEE Comp. Soc., 2003.
\newblock{\url{http://doi.ieeecomputersociety.org/10.1109/LICS.2003.1210067}}.


\bibitem{survey}
V.~Cortier, S.~Delaune, and P.~Lafourcade.
\newblock A survey of algebraic properties used in cryptographic protocols.
\newblock {\em Journal of Computer Security}, 14(1):1--43, 2006.


\bibitem{delaune}
S.~Delaune.
\newblock {\em V\'erification des protocoles cryptographiques et propri\'et\'es
 alg\'ebriques}.
\newblock PhD thesis, \'Ecole Normale Sup\'erieure de Cachan, 2006.
\newblock{\url{http://tel.archives-ouvertes.fr/tel-00132677/en/}}.


\bibitem{De2006}
S.~Delaune.
\newblock Easy Intruder Deduction Problems with Homomorphisms.
\newblock { \em Information Processing Letters}, volume 97(6), pages 213--218, 2006.
\doi{10.1016/j.ipl.2005.11.008}.

\bibitem{dolev}
D. ~Dolev and A. ~Yao.
\newblock On the security of public keys protocols.
\newblock In {\em IEEE Transactions on Information Theory}, volume 29(2), pages 198--208, 1983.
\newblock{\url{http://doi.ieeecomputersociety.org/10.1109/SFCS.1981.32}}.


\bibitem{maudeNPA07}
S.~Escobar,  C.~Meadows and J.~Meseguer.
\newblock Maude-NPA: Cryptographic Protocol Analysis Modulo Equational Properties.
\newblock In {\em Foundations of Security Analysis and Design V, FOSAD 2007/2008/2009
               Tutorial Lectures}, volume 5705 of {\em
 LNCS}, pages 1--50. Springer-Verlag, 2007.
\doi{10.1007/978-3-642-03829-7\_1}.

\bibitem{frumkin}
M. A.~Frumkin.
\newblock Polynomial time Algorithms in the Theory of Linear Diophantine Equations.
\newblock In {\em Proc. of Fundamentals of Computation Theory}, volume 56 of {\em LNCS}, pages 386--392, Springer-Verlag, 1977.
\doi{10.1007/3-540-08442-8\_106}.


\bibitem{lafourcade}
P.~Lafourcade, D.~Lugiez and R.~Treinen.
\newblock Intruder Deduction for {\it AC}-Like Equational Theories with Homomorphisms
\newblock In {\em Proc. 16$^{th}$ International Conference on Term Rewriting and Applications (RTA'05)}, volume 3467 of {\em LNCS}, pages 308--322, Springer-Verlag, 2005.
\doi{10.1007/978-3-540-32033-3\_23}.

\bibitem{Laf07}
P.~Lafourcade.
\newblock Intruder Deduction for the equational theory of exclusive-or with commutative and distributive encryption.
\newblock In {\em Electr. Notes Theor. Comput. Sci.}, volume 171(4): 37--57, 2007.
\doi{10.1016/j.entcs.2007.02.054}.

\bibitem{mcallester}
D.~McAllester.
\newblock Automatic recognition of tractability in inference relations.
\newblock {\em Journal of the ACM}, volume 40, pages 284--303, 1990.
\doi{10.1145/151261.151265}.

\bibitem{nantesayala}
D.~Nantes-Sobrinho and M.~Ayala-Rinc\'on.
\newblock Reduction of the Intruder Deduction Problem into Equational Elementary Deduction for Electronic Purse Protocols with Blind Signatures.
\newblock In {\em Proc.  17$^{th}$ Int. Workshop on Logic, Language, Information and Computation
              (WoLLIC'10)}, volume 6188 of  {\em LNCS},
pages 218--231, Springer-Verlag, 2010.
\doi{10.1007/978-3-642-13824-9\_18}.

\bibitem{papadimitriou}
C.~Papadimitriou.
\newblock{Computational Complexity.}
\newblock{Addison-Wesley, Inc}.

\bibitem{Tiu2007}
A.~Tiu.
\newblock A trace based simulation for the spi calculus: An extended abstract. 
\newblock In  {\em Proc. 5$^{th}$ Asian Symposium on Programming Languages and Systems (APLAS'07)}, volume 4807 of {\em LNCS}, pages 367--382, Springer-Verlag, 2007. 
\newblock{\url{arXiv:0901.2166}}.

\bibitem{TiGo2009}
A.~Tiu and R.~Gor\'e and J.~Dawson.
\newblock A proof theoretic analysis of intruder theories.
\newblock In {\em Proc.  20$^{th}$ International
               Conference on Rewriting Techniques and Applications (RTA'09)},
 volume 5595 of {\em LNCS}, pages 103--117.
 Springer-Verlag, 2009.
 \doi{10.2168/LMCS-6(3:12)2010}.
 
\end{thebibliography}



\end{document}